\documentclass[letter,scriptaddress,prl,notitlepage,twoside,sd]{revtex4-1}

\usepackage{amsmath,amssymb,lipsum}
\usepackage[table]{xcolor}
\usepackage{booktabs,mathrsfs,bbm,amssymb,balance,graphicx}

\DeclareMathAlphabet{\mathcal}{OMS}{cmsy}{m}{n}
\DeclareSymbolFont{Letters}{OML}{cmm}{m}{it}

\usepackage[T1]{fontenc}

\def\l{\left(}
\def\r{\right)}

\usepackage{amsthm}
\newtheorem{theorem}{Theorem}
\newtheorem{definition}{Definition}

\def\bsl{\boldsymbol{\Lambda}}
\def\bsls{\boldsymbol{\Lambda}_\mathsf{S}}
\def\bslsi{\boldsymbol{\Lambda}_\mathsf{S}^{\textrm{inv}}}
\def\PH{\Phi_{\bsls}}
\def\PHI{\Phi_{\bslsi}}

\def\SC{*_{\bsls}}
\def\SCI{*_{\bslsi}}

\def\K{K_b}
\def\Ki{K_b^*}

\def\j{{\jmath}}
\def\DE{\stackrel{\mathrm{def}}{=}}

\newcommand{\DEq}[1]{\stackrel{(\ref{#1})}{=}}

\newcommand{\T}[1]{\mathscr{T}_{\textsf{SAFT}} \left[ #1 \right]}
\newcommand{\nvec}[1]{ {\underline{#1}}}

\newcommand{\bs}[1]{ {\boldsymbol{#1}}}

\newcommand{\saft}[1]{\widehat{#1}_{\boldsymbol{\Lambda}_\mathsf{S}} \left( \omega \right) }
\newcommand{\saftn}[1]{\widehat{#1}_{\boldsymbol{\Lambda}_\mathsf{1}} \left( \omega/\sqrt{2} \right) }
\newcommand{\up}[2]{ {\overset{\lower0.5em\hbox{$\smash{\scriptscriptstyle\rightharpoonup}$}} {{#1}}} \left( {#2} \right)}
\newcommand{\dn}[2]{ {\overset{\lower0.5em\hbox{$\smash{\scriptscriptstyle  \leftharpoonup}$}} {{#1}}} \left( {#2} \right)}

\begin{document}

\title{\Large \bf Convolution and Product Theorem \\ for the Special Affine Fourier Transform}

\author{\bf{Ayush Bhandari}}
\email{ayush@MIT.edu or Ayush.Bhandari@googlemail.com}
\author{\bf{Ahmed.~I.~Zayed}}
\affiliation{Massachusetts Institute of Technology} 
\affiliation{DePaul University}


\markboth{\color{blue}\sf{Technical Report No. XX,~Vol.~XX, No.~X, Month~20XX}}%
{Bhandari and Zayed: Convolution and Product Theorem for the SAFT }

\date{\today}
\begin{abstract} The Special Affine Fourier Transform or the SAFT generalizes a number of well known unitary transformations as well as signal processing and optics related mathematical operations. Unlike the Fourier transform, the SAFT does not work well with the standard convolution operation.

Recently, Q. Xiang and K. Y. Qin  introduced a new convolution operation that is more suitable for the SAFT and by which the SAFT of the convolution of two functions is the product of their SAFTs and a phase factor. However, their convolution structure does not work well with the inverse transform in sofar as the inverse transform of the product of two functions is not equal to the convolution of the transforms.
 In this article we introduce a new convolution operation that works well with both the SAFT and its inverse leading to an analogue of the convolution and product formulas for the Fourier transform. Furthermore, we introduce a second convolution operation that leads to the elimination of the phase factor in the convolution formula obtained by Q. Xiang and K. Y. Qin.
\end{abstract}


\maketitle
\tableofcontents

\section{Introduction}

Let $\mathscr{T}: f \to \widehat{f}$ be a unitary, integral operator which maps $f$ to its transform domain representation $\widehat{f}$. For example, if $\mathscr{T}$ is the Fourier operator, that is,
\begin{equation}
\label{Fourier}
 {\mathscr{T}_\textsf{FT} \left[ f \right] } \left( \omega  \right) \DE \widehat f\left( \omega  \right) = \frac{1}{\sqrt{2\pi}}\int_\mathbb{R} {f\left( t \right){e^{ - \jmath \omega t}}dt},
\end{equation}
then $\widehat f$ is identified as the frequency domain of $f$. Furthermore, let $*$ denote the standard convolution operator defined by
$$(f*g)(t)=\frac{1}{\sqrt{2\pi}}\int_\mathbb{R}f(x)g(t-x)dx. $$ It is well known that for some functions $f$ and $g$, and $\mathscr{T} = \mathscr{T}_\textsf{FT}$, we have,
\begin{equation}
\label{CFT}
\mathscr{T}_\textsf{FT} \left[ {f * g} \right]\left( \omega  \right) = \mathscr{T}_\textsf{FT} \left[ f \right]\left( \omega  \right)\mathscr{T}_\textsf{FT} \left[ g \right]\left( \omega  \right).
\end{equation}
This result is known as the Fourier \textit{convolution theorem}.

The purpose of this paper is to extend and establish the Fourier convolution theorem for the \textbf{Special Affine Fourier Transform} (SAFT)---a phase space transform---which generalizes a number of well known transformations. Some of the interesting transformations and signal/optical operations that can be obtained from the SAFT as special cases are listed in Table~\ref{tab:1}.
\subsection{Phase Space Transformations}
Phase--Space transformations such the the fractional Fourier Transform (FrFT) \cite{Ozaktas2001} and the Linear Canonical Transform (LCT) \cite{Moshinsky} are becoming increasing popular in the areas of signal processing and communications. A remarkable feature of the \textit{phase space} transformations is that they generalize the Fourier Transformation and hence, all the mathematical developments are compatible with the Fourier analysis.

In recent years, a number of fundamental, signal processing centric theories for phase space have been developed. Some examples include convolution theorems \cite{Almeida1997, Zayed1998,Akay2001,Wei2009,Xiang2012,Bhandari2012}, sampling theory \cite{Xia1996,Stern2007, Tao2008,Tao2008a, Zhao2009,Bhandari2012, Shi2012}, time--frequency representations \cite{Capus2003,Tao2010}, shift--invariant signal approximation \cite{Bhandari2012}, sparse sampling theory \cite{Bhandari2010} and super-resolution theory \cite{Bhandari2015}.

\subsection{Convolution Theorems for the FrFT and the LCT}
In context of signal processing theory, Almeida first studied the fractional Fourier Transform (FrFT) domain representation of convolution and product operators \cite{Almeida1997}. Unfortunately, Almeida's formulation did not conform with the classical Fourier convolution--multiplication property. That is to say, the convolution of functions in time domain did not result in multiplication of their respective FrFT spectrums. As a follow up, Zayed formulated the convolution operation for the FrFT which resulted in an elegant convolution--multiplication property in FrFT domain \cite{Zayed1998}.

Recently, Xiang and Qin \cite{Xiang2012} introduced a new convolution operation that is more suitable for the SAFT and by which the SAFT of the convolution of two functions is the product of their SAFTs and a phase factor. However, their convolution structure does not work well with the inverse transform in sofar as the inverse transform of the product of two functions is not equal to the convolution of the transforms.

 In this article we introduce a new convolution operation that works well with both the SAFT and its inverse leading to an analogue of the convolution and product formulas for the Fourier transform. Furthermore, we introduce a second convolution operation that leads to the elimination of the phase factor in the convolution formula obtained in \cite{Xiang2012}.

\begin{figure*}[!t]
 \centering
\begin{equation}
\label{saftkernel}
{\kappa_{\bsls} }\left( {t,\omega } \right) = \K^*\exp \left( { - \frac{\j}{{2b}}\left( {a{t^2} + d{\omega ^2} + 2t\left( {p - \omega } \right) - 2\omega \left( {dp - bq} \right)} \right)} \right), \qquad \K = \frac{1}{{\sqrt {2\pi b} }}
\end{equation}
\hrule
\end{figure*}
\subsection{Special Affine Fourier Transform (SAFT)}
The SAFT was introduced by Abe and Sheridan \cite{Abe1994,Abe1994a} who studied a transformation in phase space that was associated with a general, inhomogeneous, lossless linear mapping. Such a transformation can model a number of optical operations such as rotation and magnification (see Table~\ref{tab:1}).

Let $f^*$ denote the complex--conjugate of $f$ and $\left\langle {f,g} \right\rangle  = \int {f\left( t \right){g^*}\left( t \right)dt}$ be the standard $L_2$ inner--product. The SAFT operation, that is, $\mathscr{T}_{\textsf{SAFT}} : f \to \widehat{f}_{\bsls}$, is defined as,
\begin{equation}
{\widehat f_{\bsls} }\left( \omega  \right) =  \begin{cases}
  {\left\langle {f,{\kappa _{\bsls} }\left( { \cdot ,\omega } \right)} \right\rangle }&{b \ne 0} \\
  {\sqrt d {e^{\j\frac{{cd}}{2}{{\left( {\omega  - p} \right)}^2} + \j\omega q}}x\left( {d\left( {\omega  - p} \right)} \right)}&{b = 0}
 \end{cases},
\label{saft}
\end{equation}
where,
\begin{itemize}
  \item $\bsls^{(2\times3)}$ is the augmented SAFT parameter matrix of form,
\begin{equation}
\label{SAFTmat}
{\bsls} = \left[ {\begin{array}{*{20}{c}}  \bsl&\vline & \nvec{\lambda  } \end{array}} \right]
\end{equation}
which is in turn parameterized by the LCT matrix $\bsl_\textsf{L}$ \cite{Moshinsky} (see Table~\ref{tab:1}) and an offset vector $\nvec{\lambda}$ such that,
\[\bs{\Lambda}  = \bigl [ {\begin{smallmatrix}
  a&b \\
  c&d
\end{smallmatrix}} \bigr] \mbox{ with } ad-bc = 1 \quad \mbox{and} \quad \nvec{\lambda}  = \bigl [ {\begin{smallmatrix}
  p \\
  q
\end{smallmatrix}} \bigr] . \]
This is the reason the SAFT is sometimes referred to as the \textit{Offset Linear Canonical Transform} or the OLCT.
\item $\kappa_{\bsls} \left(t,\omega \right)$ in (\ref{saftkernel}) is the SAFT kernel parameterized by SAFT matrix $\bsls$.
\end{itemize}

Thanks to the additive property of the SAFT/OLCT \cite{Pei2003}, the inverse--SAFT (or the iSAFT) is simply the SAFT evaluated using matrix $\bslsi$ with parameters,
\begin{equation}
\label{invmat}
\bslsi \DE  \left[ {\begin{array}{*{20}{c}}
  { + d}&{ - b}&\vline & {bq - dp} \\
  { - c}&{ + a}&\vline & {cp - aq}
\end{array}} \right] =
\left[ {\begin{array}{*{20}{c}}
  { + d}&{ - b}&\vline & {p_0} \\
  { - c}&{ + a}&\vline & {q_0}
\end{array}} \right].
\end{equation}
As a result, we are able to define the inverse transform iSAFT,
\begin{equation}
\label{iSAFT}
f\l t\r= {C_{\bslsi}}\left\langle {{{\widehat f}_{\bsls}},{\kappa _{\bslsi}}\left( { \cdot ,t} \right)} \right\rangle
\end{equation}
with some transform dependent phase constant, $$C_{\bslsi} = \exp \left( {\frac{\jmath }{2}\left( {cd{p^2} + ab{q^2} - 2adpq} \right)} \right).$$

Next, we develop a convolution structure for the SAFT denoted by $\SC$ so that we can obtain a representation of form,
$$\T{f \SC g} \propto \T{f}\T{g} $$
which is consistent with the Fourier convolution theorem (\ref{CFT}).
\section{Convolution Theorem for the SAFT}
Before we define the convolution operation in the SAFT domain, let us introduce the chirp modulation operation.

\begin{definition}[Chirp Modulation] \label{DefConv} Let $\mathbf{A}=[a_{j,k}]$ be a $2\times 2$ matrix. We define the modulation function,
\begin{equation}
\label{CP}
m_{\mathbf{A}} \l t \r \DE \exp\l \jmath \frac{ a_{1 1}}{ 2 a_{12}}  t^2 \r.
\end{equation}
Furthermore, for a given function $f,$  we define its chirp modulated functions associated with the matrix $\mathbf{A}$ as,
\begin{equation}
\label{upchirp}
\up{f}{t}  \DE   m_{\mathbf{A}} \l t \r f\l t\r \quad \mbox {and} \quad \dn{f}{t}  \DE   m^*_{\mathbf{A}} \l t \r f\l t\r.
\end{equation}
\label{CM}
\end{definition}
\noindent For example, let $\mathbf{A} = \bsls$, then, we have, $\up{f}{t} = m_{\bsls}f \l t \r = e^{\jmath \frac{a t^2}{2b}} f\l t \r $. For the case when $\mathbf{A} = \bslsi$, we get, $\up{f}{t} = m_{\bslsi}\l t \r f \l t \r = e^{-\jmath \frac{d t^2}{2b}} f\l t \r $.

Next we define the SAFT convolution operator.

\begin{definition}[SAFT Convolution/Filtering]  Let $f$ and $g$ be  two given functions and $*$ denote the usual convolution operation (see (\ref{CFT})). The SAFT convolution is defined by,
\begin{equation}
\label{SAFTconv}
h\l t \r = \l f \SC g\r \l t \r \DE \K m^*_{\bsls} \l t \r \l  \up{f}{t} * \up{g}{t} \r.
\end{equation}
\end{definition}

\begin{figure}[!t]
\begin{center}
\includegraphics[width=0.5\columnwidth]{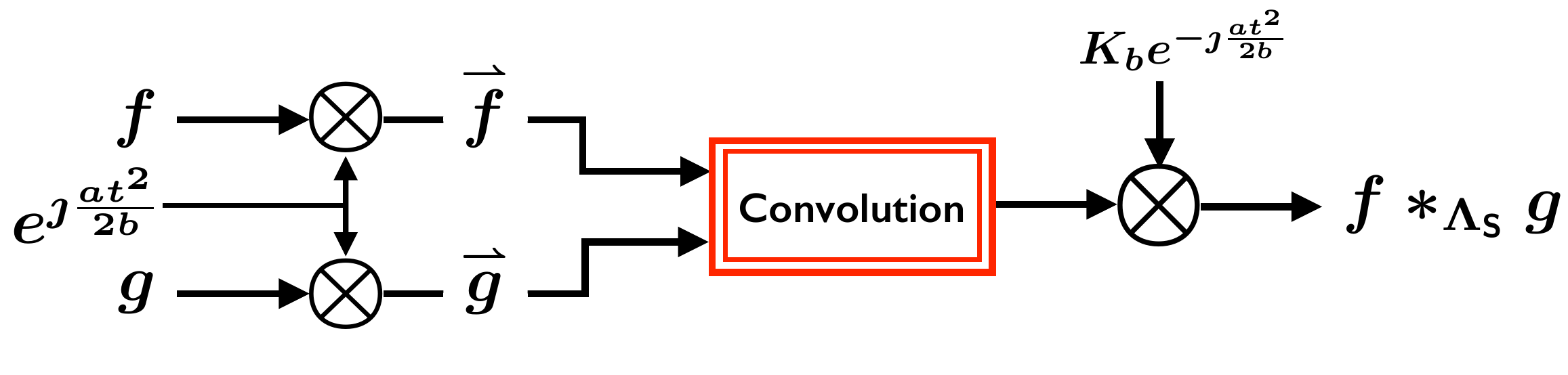}
\caption{Conceptual definition of SAFT convolution.}
\end{center}
\end{figure}

Figure~1~illustrates the block diagram for SAFT domain convolution defined in Definition~2. Next, we state the convolution and product theorem for the SAFT domain.

\begin{theorem}[SAFT Convolution and Product Theorem]
\label{thm:saftconv}
Let $f$ and $g$ be any two given functions for which the convolution $\SC$ exists and set,
$$ h\l t \r = \l f\SC g\r \l t \r. $$
Furthermore, let $\saft{f}, \saft{g}$ and $\saft{h}$ be the SAFT of $f,g$ and $h$, respectively. Then we have,
\[
h\l t \r = \l f\SC g\r  \l t \r \xrightarrow{{{\mathsf{SAFT}}}}   \saft{h} =\PH\l \omega \r \saft{f}\saft{g},
\]
where $\PH \left( \omega  \right) = {e^{\jmath \frac{\omega }{b}\left( {dp - bq} \right)}}{e^{-\jmath \frac{{d{\omega ^2}}}{{2b}}}}$. Moreover, let,
\[h\left( t \right) = \PHI \left( t \right)f\left( t \right)g\left( t \right) \mbox{ with } \PHI \left( t \right) = {e^{\jmath \frac{{a{t^2}}}{{2b}}}}{e^{ - \jmath \frac{t}{b}\left( {a{p_0} + b{q_0}} \right)}},\]
then, we have,
$\saft{h} = C_{\bslsi} \l \widehat{f} \SCI \widehat{g} \r \l \omega \r $.
%
\end{theorem}
\begin{proof}
We begin with computing the SAFT of $h$,
\begin{align*}
\saft{h} & \DEq{saft}  \T{h} \l \omega \r  = \left\langle {h\left( t \right),\kappa_{\bsls} \left( {t,\omega } \right)} \right\rangle\\
& = \int_\mathbb{R} {h\left( t \right){\kappa ^*}\left( {t,\omega } \right)dt}  \hfill \\
& \DEq{SAFTconv} \int_\mathbb{R} {\left( {\K m^*_{\bsls}\left( t \right)\int_\mathbb{R}  { \up{f}{z} \up{g}{t-z} dz} } \right){\kappa ^*}\left( {t,\omega } \right)dt} \\
&   = \underbrace{{\K^2}{e^{\jmath \frac{{d{\omega ^2}}}{{2b}}}}{e^{ - \jmath \frac{\omega }{b}\left( {dp - bq} \right)}}}_{C\l \omega \r}\int_\mathbb{R} {{e^{\jmath \frac{t}{b}\left( {p - \omega } \right)}}} \boxed{m_{\bsls}\left( t \right)} \left( {\boxed{m_{\bsls}^* \l t \r}\int_\mathbb{R} {\up{f}{z}\up{g}{t-z}dz} } \right)dt.
\end{align*}
In the above development, note that the items in the box cancel one another because $m^*_{\bsls}m_{\bsls} = 1$ (see Definition~\ref{CM}). Setting $t-z = v$ and using (\ref{upchirp}), we obtain an integral of separable form, that is, $\widehat{h}_{\bsls}\l \omega \r = {I_f}\left( \omega  \right){I_g}\left( \omega  \right)$ because,
\begin{align}
 & \underbrace{{\K^2} \PH^{*}\left( \omega  \right)\iint_\mathbb{R} {f\left( x \right)m_{\bsls}\left( x \right)g\left( v \right)m_{\bsls}\left( v \right){e^{\jmath \frac{{v + x}}{b}\left( {p - \omega } \right)}}dxdv}}_{ \widehat{h}_{\bsls}\l \omega \r} \notag \\
& \qquad \qquad \qquad \qquad \qquad \quad = {I_f}\left( \omega  \right){I_g}\left( \omega  \right) = \widehat{h}_{\bsls}\l \omega \r,
\label{sepint}
\end{align}
where, for a given function $f$, we define,
\begin{equation}
\label{IF}
{I_f}\left( \omega  \right) \DE {\K}\sqrt {\PH^{*}\left( \omega  \right)} \int_\mathbb{R} {f\left( z \right)m_{\bsls}\left( z \right){e^{\jmath \frac{z}{b}\left( {p - \omega } \right)}}dz}.
\end{equation}
Indeed, using (\ref{IF}) and (\ref{saft}), it is easy to see that,
\begin{equation}
\label{simplifiedIF}
I_f \l \omega \r = \sqrt{\PH^*\l \omega \r} \PH\l \omega \r \saft{f}
\end{equation}
and this result extends to $I_g\l \omega \r$ by symmetry.

We conclude,
\begin{align*}
&  \saft{h}  \DEq{sepint}I_f\l \omega \r I_g\l \omega \r   \\
      &  \DEq{simplifiedIF} \sqrt{\PH^*\l \omega \r} \PH\l \omega \r \saft{f} \cdot \sqrt{\PH^*\l \omega \r} \PH\l \omega \r \saft{g}\\
      & = \PH\l \omega \r \saft{f}\saft{g}
\end{align*}
which is the statement of part I of Theorem~\ref{thm:saftconv}.

Now we establish the product theorem for the SAFT,
$${\PHI}\l t \r f\l t\r g \l t \r \xrightarrow{{{\mathsf{SAFT}}}} C_{\bslsi} \l \widehat{f} \SCI \widehat{g} \r \l \omega \r.$$
Since the inverse--SAFT is the SAFT of a function with $\bsls = \bslsi$ in (\ref{invmat}), we have,
\begin{align*}
h\left( t \right)  = & C_{\bslsi}\Ki\times \\
& \int_\mathbb{R} {\saft{h}{e^{ - \jmath \frac{{\left( {a{t^2} + d{\omega ^2}} \right)}}{{2b}}}}{e^{ - \jmath \frac{\omega }{b}\left( {{p_0} - t} \right)}}{e^{\jmath \frac{t}{b}\left( {a{p_0} + b{q_0}} \right)}}d\omega }\\
= &C_{\bslsi}\Ki \PHI^*\l t \r \int_\mathbb{R} {\saft{h}{e^{ - \jmath \frac{{d{\omega ^2}}}{{2b}}}}{e^{ - \jmath \frac{{\omega \left( {{p_0} - t} \right)}}{b}}}d\omega }.
\end{align*}
By setting,
\begin{align*}
\saft{h} &= C_{\bslsi} \l \widehat{f} \SCI \widehat{g} \r \l \omega \r \\
&=C_{\bslsi} \Ki m_{\bslsi}^{*} \l \omega \r \l \up{\widehat{f}}{\omega} * \up{\widehat{g}}{\omega}\r,
\end{align*}
where, for $\bsls = \bslsi$ (see (\ref{invmat})), we have, $$m_{\bslsi} \l \cdot \r = e^{-\jmath \frac{d \l \cdot \r^2}{2b}} \mbox{ and }\up{f}{\omega} = m_{\bslsi} \l \omega \r f \l \omega \r.$$
Upon simplification, we obtain the separable integrals,
{
\begin{align*}
 h \l t \r & =   \l C_{\bslsi} \Ki \r^2 \PHI^*\l t \r  \iint \widehat{f} \l \nu \r m_{\bslsi}\l \nu \r  \widehat{g} \l \omega -  \nu \r m_{\bslsi}\l \omega - \nu \r d \nu  e^{ - \jmath \frac{{\omega \left( {{p_0} - t} \right)}}{b}} d\omega\\
& = \int   \up{\widehat{f}}{\nu}   e^{ - \jmath \frac{{\nu \left( {{p_0} - t} \right)}}{b}} d\nu
\int  \up{\widehat{g}}{\omega}  e^{ - \jmath \frac{{\omega \left( {{p_0} - t} \right)}}{b}} d\omega\\
& =  I_f\l t \r I_g\l t \r,
\end{align*}}
where,
\begin{align*}
I_f\l t \r & = C_{\bslsi} \Ki \sqrt{\PHI^*\l t \r} {\int_\mathbb{R}\up{f}{\omega}} e^{ - \jmath \frac{{\omega \left( {{p_0} - t} \right)}}{b}} d\omega. \\
& = \sqrt{\PHI^*\l t \r} \PHI \l t \r f\l t \r.
\end{align*}
As a result, we have,
\begin{align*}
 h \l t \r & =  I_f\l t \r I_g\l t \r \\
& = \sqrt{\PHI^*\l t \r} \PHI \l t \r f\l t \r \cdot \sqrt{\PHI^*\l t \r} \PHI \l t \r g \l t \r\\
& = \PHI \l t \r f \l t \r g \l t \r
\end{align*}
which is the desired result.
\end{proof}

\section{Comparison and Alternative Results}
%
%

In this section we compare our results with those of \cite{Xiang2012}. We show that our approach is not only easier to derive, but also provides more symmetric
formulas to implement. Our convolution formula is the same as the one given in  \cite{Xiang2012}, both assert that the SAFT of the convolution of two functions is the product of their SAFT
and a phase factor given by $\PH\l \omega \r$. But our product formula is different from that in \cite{Xiang2012} which states that the SAFT of the product of two functions $f$ and $g$ is,
\begin{equation}
\label{product7}
K_b^2  \PH^*\left( \omega  \right)\left( {\saft{f} \PH\left( \omega  \right)* \widehat{g}_{\boldsymbol{\Lambda}_{\mathsf{FT}}}\left( {\frac{\omega }{b}} \right)} \right).
\end{equation}

The reason our convolution and product formulas are more symmetric and simpler goes back to our definition of the chirp modulation, see Definition \ref{DefConv} , which uses the  adaptive  matrix $\bsls$ that accommodates both the forward and backward SAFT.

Furthermore, we will now derive another convolution for SAFT which eliminates the phase factor $\PH$ from the convolution formula.

\begin{definition}[Phase--free SAFT Convolution]
 Let $f$ and $g$ be two given functions and $*$ denote the usual convolution operation. The second SAFT convolution $\star$ is defined by,
\[ h( t ) =  (f \star g)( t ) \DE \sqrt{2} K_b m^*_{\bsls}\l t \r \left( \up{f}{.} * \up{g}{.}\right)(\sqrt{2}t). \]
\end{definition}
In view of this SAFT--convolution, we have the following theorem.

\begin{theorem}
Let $h( t ) =  (f \star g)( t ).$ Then, we have,
$$
\saft{h} = \saftn{f} \saftn{g},
$$
where $\widehat{f}_{\boldsymbol{\Lambda}_1}$ denotes the SAFT of $f$ with respect to the matrix ${\boldsymbol{\Lambda}_1} = \left[  \boldsymbol{\Lambda} \ \  | \ \ \nvec{\lambda}/ \sqrt{2} \right]$ (cf. (\ref{SAFTmat})).

\end{theorem}
\begin{proof}
Let $\Omega  = bq - dp$. We have,
\begin{eqnarray*}
\saft{h}&=& \sqrt{2} K^2_b e^{-\frac{\j at^2}{2b}} \int_{\mathbb{R}}e^{\frac{\j}{2b} \l at^2+d\omega^2 +2tp-2t\omega +2\omega \Omega \r}dt\\
&\times & \int_{\mathbb{R}}f(\tau)e^{\frac{\j at^2}{2b}}g(\sqrt{2}t-\tau)e^{\frac{\j a}{2b}(\sqrt{2}t-\tau)^2}d\tau.
\end{eqnarray*}
Setting $x=\sqrt{2}t-\tau$ and simplifying the integrals, we obtain,
\begin{align*}
\saft{h}&= K^2_b \PH^* (\omega) \int_{\mathbb{R}}e^{\frac{\j a\tau^2}{2b}}f(\tau) d\tau   \int_{\mathbb{R}} g(x)
e^{ \frac{\j }{2b} \l ax^2+ 2(p-\omega)(x+\tau)/\sqrt{2}\r } dx\\
&=    K^2_b \PH^* (\omega) \int_{\mathbb{R}}e^{ \frac{\j }{2b} \l a\tau^2+ 2(p-\omega)\tau /\sqrt{2}\r }f(\tau) d\tau    \int_{\mathbb{R}} g(x)
e^{ \frac{\j }{2b} \l ax^2+ 2(p-\omega)x/\sqrt{2}\r } dx    \\
&=   K^2_b \PH^* (\omega) \int_{\mathbb{R}}  e^{ \frac{\j }{2b} \l a\tau^2+ \sqrt{2}p\tau -\sqrt{2}\tau\omega \r }f(\tau) d\tau \int_{\mathbb{R}}  e^{ \frac{\j }{2b} \l ax^2+ \sqrt{2}px -\sqrt{2}x\omega) \r }g(x) dx.
\end{align*}
But since, $$ \PH^* (\omega)= {e^{\jmath \frac{1}{{2b}}\left( {2d{{\left( {\frac{\omega }{{\sqrt 2 }}} \right)}^2} + 2\sqrt 2 \Omega \frac{\omega }{{\sqrt 2 }}} \right)}},$$
it follows, that $\saft{h}=I_f(\omega) I_g(\omega)$,
where
{ $$ I_f(\omega)= K_b  \int_{\mathbb{R}}  e^{ \frac{\j }{2b} \l a\tau^2+ d\left(\frac{\omega}{\sqrt{2}}\right)^2+\sqrt{2}p\tau -\sqrt{2}\tau\omega
+ \sqrt{2}\Omega(\omega/\sqrt{2})\right]}f(\tau) d\tau
$$} and similar expression for $I_g(\omega).$ But it is easy to see that
$$ I_f(\omega)=\saftn{f}, $$ and this completes the proof.
\end{proof}
%
%
%


\noindent \textbf{Relation to Convolution theory of LCTs:} In the special case where $p=q=0 \Leftrightarrow \bsls = \boldsymbol{\Lambda}_{\mathsf{LCT}}$, the SAFT reduces to the LCT and the last convolution theorem takes the simple form
$$\mathscr{T}_\textsf{LCT}[f\star g](\omega)=\mathscr{T}_\textsf{LCT}[f]\left(\omega/\sqrt{2}\right)\mathscr{T}_\textsf{LCT}[g]\left(\omega/\sqrt{2}\right).$$
\section{Conclusion}
In this letter, we introduced two definitions of convolution operation that establish the convolution--product theorem for the Special Affine Fourier Transform (SAFT) introduced by Abe and Sheridan \cite{Abe1994,Abe1994a}. Our result is quite general in that our convolution--product theorem is applicable to all the listed unitary transformations in Table~\ref{tab:1}. Furthermore, we also presented a product theorem for the SAFT which establishes the fact that the product of functions in time amounts to convolution in SAFT domain. We conclude that our construction of the convolution structure for the SAFT domain establishes the SAFT duality principle, that is, convolution in one domain amounts to multiplication in the transform domain and vice--versa. Our results can be used to develop the semi--discrete convolution structure \cite{Bhandari2012} for sampling and approximation theory linked with the Special Affine Fourier Transform, but this will be done in a separate project. \\

\newpage
{\rm 
\begin{table}[t]
\centering
\caption{ \textsf{SAFT, Unitary Transformations and Operations}}
\begin{tabular*}{0.57\textwidth}{p{4cm} p{6cm}}
\toprule
\addlinespace
%
\rowcolor{blue!10} 
\hline
SAFT Parameters $\left(\bsls \right) $ & Corresponding Unitary Transform \\
\hline
\addlinespace
$\bigl[ \begin{smallmatrix} a &b  & \vline & & {0}   \\
 c& d & \vline& &{0} \end{smallmatrix} \bigr] = \pmb\Lambda_\textsf{L} $				&  \textbf{Linear Canonical Transform} \\	[2.5pt]

$\bigl[ \begin{smallmatrix} &\cos\theta&\sin\theta  & \vline & & {p}   \\
 -&\sin\theta&\cos\theta & \vline& &{q} \end{smallmatrix} \bigr] = \pmb\Lambda_\theta^O $				&  \textbf{Offset Fractional Fourier Transform} \\	[2.5pt] 

$\bigl[ \begin{smallmatrix} &\cos\theta&\sin\theta  & \vline & & {0}   \\
 -&\sin\theta&\cos\theta & \vline& &{0} \end{smallmatrix} \bigr] = \pmb\Lambda_\theta $				&  \textbf{Fractional Fourier Transform} \\	[2.5pt] 

$\bigl[ \begin{smallmatrix} &0 &1& \vline & & {p}  \\-&1&0 & \vline & & {q} \end{smallmatrix} \bigr] 
= \pmb\Lambda_\text{FT}^O$  													& \textbf{Offset Fourier Transform (FT)}  \\		[2.5pt] 

$\bigl[ \begin{smallmatrix} &0 &1& \vline & & {0}  \\-&1&0 & \vline & & {0} \end{smallmatrix} \bigr] 
= \pmb\Lambda_\text{FT}$  													& \textbf{Fourier Transform (FT)}  \\		[2.5pt] 
$\bigl[ \begin{smallmatrix} 0 & \j & \vline & & {0} \\ \j & 0 & \vline & & {0} \end{smallmatrix} \bigr] 
= \pmb\Lambda_\text{LT}$ 													& \textbf{Laplace Transform (LT)}  \\	[2.5pt] 
$\bigl[ \begin{smallmatrix} \j \cos\theta &\j \sin\theta & \vline & & {0}  \\
 \j \sin \theta & -\j\cos\theta & \vline & & {0} \end{smallmatrix} \bigr] $  								& \textbf{Fractional Laplace Transform}   \\ 	[2.5pt] 
$\bigl[ \begin{smallmatrix} 1 & b & \vline & & {0}  \\0&1 & \vline & & {0} \end{smallmatrix} \bigr]$  						& \textbf{Fresnel Transform}  \\	[2.5pt] 
$\bigl[ \begin{smallmatrix} 1&\jmath b & \vline & & {0}  \\  \jmath&1 & \vline & & {0} \end{smallmatrix} \bigr]$				& \textbf{Bilateral Laplace Transform}  \\	[2.5pt] 
$\bigl[ \begin{smallmatrix} 1&-\jmath b & \vline & & {0}  \\ 0&1 & \vline & & {0} \end{smallmatrix} \bigr]$, $b \ge 0$  		& \textbf{Gauss--Weierstrass Transform}   \\	[2.5pt] 
$\tfrac{1}{{\sqrt 2 }} \bigl[\begin{smallmatrix} 0 & e^{ - {{\jmath\pi } 
\mathord{\left/{\vphantom {{j\pi } 2}} \right.\kern-\nulldelimiterspace} 2}} & \vline & & {0} 
\\-e^{ - {{\jmath\pi } \mathord{\left/{\vphantom {{j\pi } 2}} \right.\kern-\nulldelimiterspace} 2}} 
&1 & \vline & & {0} \end{smallmatrix} \bigr]$  													& \textbf{Bargmann Transform} \\[2.5pt] 
%
\addlinespace
\hline
\rowcolor{blue!10} 
SAFT Parameters $\left(\bsls \right) $ & Corresponding Signal Operation \\
\hline
\addlinespace
$\bigl[ \begin{smallmatrix} 1/\alpha& 0  & \vline & & {0}   \\
 0 & \alpha & \vline& &{0} \end{smallmatrix} \bigr] = \pmb\Lambda_\alpha $				&  \textbf{Time Scaling} \\	[2.5pt] 

$\bigl[ \begin{smallmatrix} 1 & 0  & \vline & & {\tau}   \\
 0 & 1 & \vline& &{0} \end{smallmatrix} \bigr] = \pmb\Lambda_\tau $				&  \textbf{Time Shift} \\	[2.5pt]

$\bigl[ \begin{smallmatrix} 1 & 0  & \vline & & {0}   \\
 0 & 1 & \vline& &{\xi} \end{smallmatrix} \bigr] = \pmb\Lambda_\xi $				&  \textbf{Frequency Shift} \\	[2.5pt] 

\addlinespace
\hline
\rowcolor{blue!10} 
SAFT Parameters $\left(\bsls \right) $ & Corresponding Optical Operation \\
\hline
\addlinespace
$\bigl[ \begin{smallmatrix} &\cos\theta&\sin\theta  & \vline & & {0}   \\
 -&\sin\theta&\cos\theta & \vline& &{0} \end{smallmatrix} \bigr] = \pmb\Lambda_\theta $			&  \textbf{Rotation} \\	[2.5pt] 

$\bigl[ \begin{smallmatrix} 1 & 0  & \vline & & {0}   \\
 \tau & 1 & \vline& &{0} \end{smallmatrix} \bigr] = \pmb\Lambda_\tau $				&  \textbf{Lens Transformation} \\	[2.5pt] 
  
 $\bigl[ \begin{smallmatrix} 1 & \eta  & \vline & & {0}   \\
 0 & 1 & \vline& &{0} \end{smallmatrix} \bigr] = \pmb\Lambda_\eta $				&  \textbf{Free Space Propagation} \\	[2.5pt] 
 
 $\bigl[ \begin{smallmatrix} e^{\beta} & 0  & \vline & & {0}   \\
 0 & e^{-\beta}  & \vline& &{0} \end{smallmatrix} \bigr] = \pmb\Lambda_\beta $				&  \textbf{Magnification} \\	[2.5pt] 
 
  $\bigl[ \begin{smallmatrix} \cosh\alpha & \sinh\alpha  & \vline & & {0}   \\
\sinh\alpha & \cosh\alpha  & \vline& &{0} \end{smallmatrix} \bigr] = \pmb\Lambda_\eta $				&  \textbf{Hyperbolic Transformation} \\	[2.5pt] 

\bottomrule
\end{tabular*}
\label{tab:1}
\end{table}}

\bibliographystyle{IEEEtran}
\bibliography{SAFTShannon.bib}

\end{document}